\documentclass[5p,times]{elsarticle}
\journal{European Journal of Control}

\usepackage{graphicx}
\usepackage{subcaption}
\usepackage{etoolbox} 
\usepackage{scalerel} 
\usepackage{siunitx}
\usepackage{pifont}
\usepackage{graphicx}
\usepackage{amsmath,amssymb,textcomp}
\usepackage{algorithm} 
\usepackage{cite}
\usepackage{booktabs}
\usepackage{url}
\usepackage{algpseudocode}
\usepackage{paralist}
\usepackage{color}
\usepackage[T1]{fontenc}
\usepackage{stmaryrd}
\usepackage{epsfig} 
\usepackage{bm} 
\usepackage{listings,lstautogobble}
\usepackage{mathtools}
\usepackage{mathrsfs}
\usepackage{verbatim} 
\usepackage{epstopdf}
\usepackage{listings}
\usepackage{parcolumns}
\usepackage{color}
\usepackage{mathtools}
\usepackage{mathrsfs}
\usepackage{xspace}
\usepackage{verbatim}
\usepackage[utf8]{inputenc}
\usepackage{calrsfs}
\usepackage{rotating,multirow}
\usepackage{mathtools}

\usepackage[most]{tcolorbox}

\usepackage[T1]{fontenc}
\usepackage[utf8]{inputenc}

\usepackage{amsthm}
\newtheorem{definition}{Definition}

\newtheorem{theorem}{Theorem}
\newtheorem{remark}{Remark}
 \newtheorem{lemma}{Lemma} 
 \newtheorem{assumption}{Assumption}

  \renewcommand{\footnotesize}{\fontsize{8pt}{11pt}\selectfont}

\DeclareMathAlphabet{\mathcal}{OMS}{cmsy}{m}{n}

\DeclarePairedDelimiter{\abs}{\lvert}{\rvert}
\DeclarePairedDelimiter{\norm}{\lVert}{\rVert}

\usepackage{hyperref}

\def \tb{\textcolor{black}}

\newcommand{\zono}[1]{\langle #1 \rangle}

\allowdisplaybreaks

\newcommand{\sys}{ZPC}

\begin{document}

\begin{frontmatter}

\title{Robust Data-Driven Predictive Control using Reachability Analysis}
\author[add1,add2]{Amr Alanwar\corref{cor1}}
\ead{alanwar@kth.se}
\author[add2,add3]{Yvonne Stürz\corref{cor1}}
\ead{stuerz@kth.se}
\author[add2]{Karl Henrik Johansson}
\ead{kallej@kth.se}

\address[add1]{Computer Science \& Electrical Engineering Department, Jacobs University Bremen, Germany}
\address[add2]{Division of Decision and Control Systems, KTH Royal Institute of Technology, Sweden}
\address[add3]{Model Predictive Control Laboratory, University of California, Berkeley, USA}

\cortext[cor1]{Correspondence authors which are with equal contributions.}





\begin{abstract}

We present a robust data-driven control scheme for an unknown linear system model with bounded process and measurement noise. Instead of depending on a system model in traditional predictive control, a controller utilizing data-driven reachable regions is proposed. The data-driven reachable regions are based on a matrix zonotope recursion and are computed based on only noisy input-output data of a trajectory of the system. We assume that measurement and process noise are contained in bounded sets. While we assume knowledge of these bounds, no knowledge about the statistical properties of the noise is assumed. 
In the noise-free case, we prove that the presented purely data-driven control scheme results in an equivalent closed-loop behavior to a nominal model predictive control scheme. 
In the case of measurement and process noise, our proposed scheme guarantees robust constraint satisfaction, which is essential in safety-critical applications. Numerical experiments show the effectiveness of the proposed data-driven controller in comparison to model-based control schemes. 
%
\end{abstract}

\end{frontmatter}

\section{Introduction}

Model predictive control (MPC) is a well-established control method that can handle input and state constraints \citep{conf:mpcbook}. It optimizes the control performance for a given cost function over the system's predicted evolution. In order to implement traditional MPC, a model of the system is thus required. System identification and accurate modeling can be very costly for complex systems, such as robotics applications, or even impossible, such as fluid dynamics \tb{\citep{conf:deepc,conf:deepctherom}}. Learning-based approaches have therefore been investigated to learn a model from data. Most of these methods, however, are data-intensive and do not provide safety guarantees in general. Furthermore, while robust MPC is an active field of research, it is still hard to account for parametric model mismatch and process, or measurement noise~\tb{\citep{robustMPC1}}. 
Therefore, this paper focuses on robust data-driven predictive control for unknown linear systems under measurement and process noise using reachability analysis. 

Many data-driven methods can be mentioned in relevance to our approach. One category employs \tb{the fundamental lemma}, which views a dynamical system by the subspace of the signal space in which the system's trajectories appear \citep{conf:willems,conf:deepc,conf:mpcguarantees,conf:newpersp}. Recent work has utilized the fundamental lemma \citep{conf:willems} in MPC~\citep{conf:deepc}. Moreover, the authors in \citep{conf:mpcguarantees} provide data-driven MPC with stability and robustness guarantees. Also, data-driven feedback controllers and stabilization are discussed in \citep{conf:formulas,conf:nonexciting,conf:robust,conf:nonlinear1}.  
Recent developments in the data-driven direction include robust controller synthesis from noisy input-state trajectories \citep{conf:robustcontrol} and data-driven optimal control \citep{conf:optimalcont1,conf:optimalcont2}. Another category computes the set of possible models given the available data and then derives controller and system properties for the computed set  \citep{conf:datadriven_reach,conf:data_Driven_Estimation,conf:firstset_dissipativity1,conf:dissipativity2}.

%
%
%
%
  %
%
%
%
  %

Reachability analysis computes the set of states that a system can reach within finite or infinite time when starting from a bounded set of initial states, subject to a set of possible inputs \citep{conf:reach2000}. Many research branches utilize reachability analysis, such as formal verification, formal controller synthesis, set-based estimation, and set-based prediction \citep{conf:reviewAlthoff}. The most popular approaches in computing reachable sets are set propagation and simulation-based techniques. The efficiency of propagation-based methods depends on the set representation: polyhedra \citep{conf:reachpolydt}, zonotopes \citep{conf:thesisalthoff}, (sparse) polynomial zonotopes \citep{conf:sparsepolyzono}, ellipsoids \citep{conf:reachellipsoidal}, support functions \citep{conf:reachsupportlinear}, and Taylor series \citep{conf:rigorousreachtaylor}. The zonotopes have favorable properties as they can be represented compactly, and they are closed under the Minkowski sum and linear mapping.

  


This paper considers data-driven predictive control by computing the set of models consistent with noisy data, and it is considered the first step in this track. Our proposed approach consists of two phases: the data-collection phase and the control phase. During the data-collection phase, we collect input and output data samples from the unknown system. The collected data is used to compute an implicit data-driven system representation using matrix zonotopes, which is based on ideas from \citep{conf:datadriven_reach}. During the control phase, we employ a zonotopic data-driven predictive control scheme (ZPC). In particular, \sys{} computes the data-driven reachable set based on a matrix zonotope recursion starting from the measured output $y(t)$. The matrix zonotope recursion utilizes the learned data-driven system representation. The optimal control problem that is solved by \sys{} during the control phase results in the optimal input $u(t)$ that minimizes a given cost function such that the output $y(t)$ stays within the predicted reachable set and the output constraints are robustly satisfied. The code to recreate our findings is publicly available\footnotemark.
The contributions of this paper can be summarized as follows: 
\footnotetext{\href{https://github.com/aalanwar/Data-Driven-Predictive-Control}{https://github.com/aalanwar/Data-Driven-Predictive-Control}}
\begin{enumerate}
    \item We propose a robust data-driven predictive control scheme. In the first phase, noisy data is collected from the system with unknown model. A single input-output trajectory can be sufficient. A matrix zonotope recursion is used as a data-driven reachability prediction within a predictive control scheme based on the collected noisy data. 
    \item In the noise-free case, we prove that the proposed data-driven \sys{} scheme results in an equivalent closed-loop performance as a nominal MPC scheme. 
    \item Under measurement and process noise, we guarantee robust constraint satisfaction of the closed-loop system under the feasibility of the proposed data-driven predictive control scheme at each time step. 
\end{enumerate}


The rest of the paper is organized as follows: Section~\ref{sec:pb} gives the problem statement and provides relevant preliminaries. The proposed data-driven predictive control is presented in Section~\ref{sec:solu}. Section~\ref{sec:eval} shows the evaluation of the proposed algorithm in numerical experiments. Finally, Section~\ref{sec:conc} concludes the paper. 




\section{Preliminaries and Problem Statement}\label{sec:pb}

We start by defining some preliminaries before stating the problem setting. 

\subsection{Set Representations}

\begin{definition}(\textbf{Zonotope} \citep{conf:zono1998}) \label{def:zonotopes} 
Given a center $c_\mathcal{Z} \in \mathbb{R}^n$ and a number $\gamma_\mathcal{Z} \in \mathbb{N}$
of generator vectors in the generator matrix $G_\mathcal{Z}=\begin{bmatrix} g^{(1)}_\mathcal{Z},\dots ,g^{(\gamma_\mathcal{Z})}_\mathcal{Z}\end{bmatrix} \in \mathbb{R}^{n \times \gamma_\mathcal{Z}}$, a zonotope is defined as 
\begin{equation}
	\mathcal{Z} {=} \Big\{ x \in \mathbb{R}^n \; \Big| \; x = c_\mathcal{Z} + \sum_{i=1}^{\gamma_\mathcal{Z}} \beta^{(i)} \, g^{(i)}_\mathcal{Z} \, ,
	-1 \leq \beta^{(i)} \leq 1 \Big\}.
\end{equation}
Furthermore, we define the shorthand $\mathcal{Z} = \zono{c_\mathcal{Z},G_\mathcal{Z}}$. 
\end{definition}

\begin{definition} (\textbf{Matrix Zonotope} \citep[p. 52]{conf:thesisalthoff}) \label{def:matzonotopes} 
Given a center matrix $C_\mathcal{M} \in \mathbb{R}^{n\times j}$ and a number $\gamma_\mathcal{M} \in \mathbb{N}$ of generator matrices $G_\mathcal{M}=[\tilde{G}^{(1)}_\mathcal{M},\dots,\tilde{G}^{(\gamma_\mathcal{M})}_\mathcal{M}] \in \mathbb{R}^{n \times \gamma_\mathcal{M} j}$, a matrix zonotope is defined by 
\begin{equation}
	\mathcal{M} {=} \Big\{ X \in \mathbb{R}^{n\times k} \; \Big| \; X = C_\mathcal{M} + \sum_{i=1}^{\gamma_\mathcal{M}} \beta^{(i)} \, \tilde{G}^{(i)}_\mathcal{M} \, ,
	-1 \leq \beta^{(i)} \leq 1 \Big\}.
\end{equation}
Furthermore, we define the shorthand $\mathcal{M} = \zono{C_\mathcal{M},G_\mathcal{M}}$. 
\end{definition}

The linear map is defined and computed as follows \citep{conf:thesisalthoff}:
\begin{align}
L \mathcal{Z} = \{Lz\, | \, z\in\mathcal{Z}\}  = \zono{L c_{\mathcal{Z}}, L G_{\mathcal{Z}} }. \label{eq:linmap}
\end{align} 
Given two zonotopes $\mathcal{Z}_1=\langle c_{\mathcal{Z}_1},G_{\mathcal{Z}_1} \rangle$ and $\mathcal{Z}_2=\langle c_{\mathcal{Z}_2},G_{\mathcal{Z}_2} \rangle$, the Minkowski sum $\mathcal{Z}_1 \oplus \mathcal{Z}_2 = \{z_1 + z_2 \, | \, z_1\in \mathcal{Z}_1, z_2 \in \mathcal{Z}_2 \}$ can be computed exactly as \citep{conf:thesisalthoff}: 
\begin{equation}
     \mathcal{Z}_1 \oplus \mathcal{Z}_2 = \Big\langle c_{\mathcal{Z}_1} + c_{\mathcal{Z}_2}, [G_{\mathcal{Z}_1}, G_{\mathcal{Z}_2} ]\Big\rangle.
     \label{eq:minkowski}
\end{equation}
For simplicity, we use the notation $+$ instead of $\oplus$ for Minkowski sum as the type can be determined from the context. Similarly, we use  $\mathcal{Z}_1 - \mathcal{Z}_2$ to denote $\mathcal{Z}_1 + -1 \mathcal{Z}_2$ \tb{not the Minkowski difference}. To over-approximate a zonotope $\mathcal{Z}=\zono{c_\mathcal{Z},\begin{bmatrix} g_{\mathcal{Z}}^{(1)} ,\dots , g_{\mathcal{Z}}^{(\gamma)} \end{bmatrix}}$ by an interval $\mathcal{V} = \begin{bmatrix} \underline{v},\bar{v}  \end{bmatrix}$, we do the following:
\begin{align}
    \bar{v} &= c_\mathcal{Z} + \sum_{i=1}^{\gamma_\mathcal{Z}} \abs{g_{\mathcal{Z}}^{(i)}} \label{eq:zonotointU}\\
    \underline{v} &= c_\mathcal{Z} -  \sum_{i=1}^{\gamma_\mathcal{Z}} \abs{g_{\mathcal{Z}}^{(i)}}\label{eq:zonotointL}
\end{align}
The Cartesian product of two zonotopes $\mathcal{Z}_1 $ and $\mathcal{Z}_2$ is defined and computed as 
\begin{align*}
\mathcal{Z}_1 \times \mathcal{Z}_2 &= \Bigg\{ \begin{bmatrix}z_1 \\ z_2\end{bmatrix} \Bigg|\, z_1 \in \mathcal{Z}_1, z_2 \in \mathcal{Z}_2 \Bigg\}= \bigg\langle \begin{bmatrix} c_{\mathcal{Z}_1}  \\ c_{\mathcal{Z}_2}  \end{bmatrix}, \begin{bmatrix} G_{\mathcal{Z}_1}  & 0 \\ 0 & G_{\mathcal{Z}_2}  \end{bmatrix} \bigg\rangle.
\label{eq:cardprod}
\end{align*}

\begin{definition} \label{def:intmat}(\textbf{Interval Matrix} \citep[p. 42]{conf:thesisalthoff})
An interval matrix $\mathcal{I}$ specifies the interval of all possible values for each matrix element between the left limit $\underline{I}$ and right limit $\bar{I}$:
\begin{align}
    \mathcal{I} = \begin{bmatrix} \underline{I},\bar{I}  \end{bmatrix}, \quad \underline{I},\bar{I} \in \mathbb{R}^{r \times c}
\end{align}
\end{definition}

\subsection{Problem Statement}
We consider a controllable discrete-time linear system
\begin{equation}
    \begin{split}
        x(t+1) &= A x(t) + B u(t) +w(t),\\
        y(t) &=  C x(t) + v(t),
    \end{split}
    \label{eq:sys}
\end{equation}
%
with the system matrices $A \in \mathbb{R}^{n \times n}$ and $B \in \mathbb{R}^{n \times m}$, $C \in \mathbb{R}^{p \times n}$, 
state $x(t) \in \mathbb{R}^{n}$, and input $u(t) \in \mathbb{R}^{m}$. We assume that the states of the system are measurable, i.e., the system output matrix is given by $C=I$, and thus the measured output is $y(t) \in \mathbb{R}^n$. 
The input and output constraints are given by 
\begin{equation}
    \begin{split}
    u(t) \in \mathcal{U}_t \subset \mathbb{R}^{m}, \\
    y(t) \in \mathcal{Y}_t \subset \mathbb{R}^n. 
    \end{split}
    \label{eq:sys_con}
\end{equation}
We assume that the process and measurement noise $w(t)$ and $v(t)$ are bounded as follows: 
\begin{assumption}\label{as:noisew}
We assume that the process noise $w(t)$ is bounded by a zonotope $w(t) \in \mathcal{Z}_w = \zono{c_{\mathcal{Z}_w},G_{\mathcal{Z}_w}}$ for all time steps. 
\end{assumption}
\begin{assumption}\label{as:noisev}
We assume that the measurement noise $v(t)$ is bounded by a zonotope $v(t) \in \mathcal{Z}_v = \zono{c_{\mathcal{Z}_v},G_{\mathcal{Z}_v}}$ for all time steps. Furthermore, we assume that the one-step propagation $A v(t)$ is bounded by a zonotope $Av(t) \in \mathcal{Z}_{Av} = \zono{c_{\mathcal{Z}_{Av}},G_{\mathcal{Z}_{Av}}}$ for all time steps similar to the assumption in \citep{conf:formulas}.
\end{assumption}

We aim to solve a receding horizon optimal control problem when the model of the system in \eqref{eq:sys} is unknown, but input and noisy output trajectories are available.

\subsection{Input-Output Data and Reachability}
 We consider $K$ input-output trajectories of different lengths $T_i$, $i=1,\dots,K$, denoted by $\{\tilde{u}^{(i)}(t)\}_{t=0}^{T_i - 1}$ and $\{\tilde{y}^{(i)}(t)\}_{t=0}^{T_i}$, $i=1, \dots, K$. We collect the set of all data sequences in the following matrices
 \begin{align*}
     Y &{=}  \begin{bmatrix} \tilde{y}^{(1)}(0)  \dots \tilde{y}^{(1)}(T_1)   \dots \; \tilde{y}^{(K)}(0)  \dots  \tilde{y}^{(K)}(T_K)\end{bmatrix}, \\
     U_- &{=} \begin{bmatrix} \tilde{u}^{(1)}(0) \dots  \tilde{u}^{(1)}(T_1-1) \dots  \tilde{u}^{(K)}(0)  \dots  \tilde{u}^{(K)}(T_K-1) \end{bmatrix}.
 \end{align*}
 Let us further denote
 \begin{align*}
     Y_+ &{=} \begin{bmatrix} \tilde{y}^{(1)}(1)  \dots  \tilde{y}^{(1)}(T_1)  \dots  \tilde{y}^{(K)}(1)  \dots  \tilde{y}^{(K)}(T_K) \end{bmatrix}, \nonumber\\
     Y_- &{=} \begin{bmatrix} \tilde{y}^{(1)}(0)  \dots  \tilde{y}^{(1)}(T_1{-}1)  \dots  \tilde{y}^{(K)}(0) \dots  \tilde{y}^{(K)}(T_K{-}1) \end{bmatrix}.
 \end{align*}
The total amount of data points from all available trajectories is denoted by $T = \sum_{i=1}^{K} T_i$ and we denote the set of all available data by $D=\{U_-, Y\}$. 

Reachability analysis in general computes the set of $y$ which can be reached given a set of uncertain initial states $\mathcal{R}_0 \subset \mathbb{R}^n$
containing the initial output $y(0) \in \mathcal{R}_0$ and a set of uncertain inputs
$\mathcal{Z}_{u,t} \subset \mathbb{R}^m$ containing the inputs $u(t) \in \mathcal{Z}_{u,t} $.

\begin{definition} 
The reachable set $\mathcal{R}_{t}$ after $N$ time steps, inputs $u(t) \in \mathcal{Z}_{u,t}  \subset \mathbb{R}^m$, $\forall t \in \{0,\dots, N-1 \}$, noise $w( \cdot) \in \mathcal{Z}_w$, and initial set $\mathcal{R}_0 \in  \mathbb{R}^n$, is the set of all state trajectories starting in $\mathcal{R}_0$ after $N$ steps: 
\begin{align}\label{eq:R}
        \mathcal{R}_{N} =& \big\{ y(N) \in \mathbb{R}^n \, \big| x(t+1) = Ax(t) {+} Bu(t) + w(t),\nonumber\\ 
        \,&  y(t) = x(t) +v(t), y(0) \in \mathcal{R}_0, w(t) \in \mathcal{Z}_w,  \nonumber \\
        \,& u(t) \in \mathcal{Z}_{u,t}: \forall t \in \{0,...,N\}\big\}.
\end{align}
\end{definition} 

Note that we define the reachable sets with respect to the output given that $C=I$ and to align with the output reachable regions computed in the coming section. 
While the actual noise sequence in the data, denoted by $\tilde{w}^{(i)} (t)$ for trajectory $i$, is unknown, we assume to know a bound on the noise as specified in Assumption~\ref{as:noisew}. 
From Assumption~\ref{as:noisew}, it follows directly that the stacked matrix 
\begin{align*}
    W_- {=} \begin{bmatrix} 
    \tilde{w}^{(1)}(0)  \dots \tilde{w}^{(1)}(T_1{-}1)  \dots \tilde{w}^{(K)}(0)  \dots  \tilde{w}^{(K)}(T_K{-}1)\end{bmatrix}
\end{align*}
is an element of the set $W_- \in \mathcal{M}_w$ where $\mathcal{M}_w =\zono{ C_{\mathcal{M}_w}, [G_{\mathcal{M}_w}^{(1)},$ $\dots,G_{\mathcal{M}_w}^{(\gamma_{\mathcal{Z}_w} T)}] }$ is the matrix zonotope resulting from the concatenation of multiple noise zonotopes $\mathcal{Z}_w=\zono{c_{\mathcal{Z}_w},[g_{\mathcal{Z}_w}^{(1)}, \dots, g_{\mathcal{Z}_w}^{(\gamma_{\mathcal{Z}_w})}]}$ as described in \citep{conf:datadriven_reach}. Similarly, we define
\begin{align*}
    V_- &= \begin{bmatrix} 
    \tilde{v}^{(1)}(0)  \dots \tilde{v}^{(1)}(T_1{-}1)  \dots \tilde{v}^{(K)}(0)  \dots  \tilde{v}^{(K)}(T_K{-}1)\end{bmatrix}, \\
    V_+ &= \begin{bmatrix} 
    \tilde{v}^{(1)}(1)  \dots \tilde{v}^{(1)}(T_1)  \dots \tilde{v}^{(K)}(1)  \dots  \tilde{v}^{(K)}(T_K)\end{bmatrix}, 
\end{align*}
which are bounded as follow: $V_-,V_+ \in \mathcal{M}_v$ and $AV_- \in \mathcal{M}_{Av}$ where $\mathcal{M}_v$ and $\mathcal{M}_{Av}$ result from the concatenation of the zonotopes $\mathcal{Z}_v$ and $\mathcal{Z}_{Av}$, respectively and are defined similar to $\mathcal{M}_w$. 


We also denote the Hankel matrix associated to vector $z$ by $\mathcal{H}_{i,j,M}(z)$, where $i$ denotes the index of the first sample, $j$ the number of block rows, and $M$ the number of block columns. 
\begin{align}
\mathcal{H}_{i,j,M}(z) {=}  \begin{bmatrix} z(i) & z(i+1) &\!\! \dots\!\!& z(i+M-1)\\
z(i+1) & z(i+2)&\!\! \dots\!\! & z(i+M)\\
\vdots & \vdots &\!\! \ddots\!\! & \vdots \\
z(i+j-1) &z(i+j) & \!\!\dots \!\!& z(i+j+M-2)
\end{bmatrix}.
\end{align}

\begin{definition}(\citep{conf:formulas})
The signal $U_-  \in \mathbb{R}^{m \times T}$
is persistently exciting of order $L$ if the matrix $\mathcal{H}_{0,L,T-L+1}(\tilde{u})$ has full rank $mL$
where $T \geq (m+1)L-1$ for the deterministic system of \eqref{eq:sys}.
\end{definition}
%
\begin{lemma}(\citep[Cor. 2]{conf:willems})
\label{lm:exciting}
If the input $U_- \in \mathbb{R}^{m \times T}$ 
is persistently exciting of order $n+k$
for the deterministic system of \eqref{eq:sys}, then 
\begin{align}
   \text{rank}\begin{bmatrix} \mathcal{H}_{0,k,T-k+1}(\tilde{y}) \\  \hline \mathcal{H}_{0,1,T-k+1}(\tilde{u})\end{bmatrix} = n + km \, .
\label{eq:persiscond}
\end{align}
\end{lemma}

A special case of Lemma~\ref{lm:exciting} for $k=1$ yields 
\begin{align}
   \text{rank}\begin{bmatrix} Y_-  \\ U_-  \end{bmatrix} = n + m \, .
\label{eq:persiscondteq1}
\end{align}

\begin{lemma}(\citep[Th. 1]{conf:formulas})
\label{lm:xkp1}
If $\text{rank}\begin{bmatrix} Y_- \\ 
U_-\end{bmatrix} = n + m$ for the deterministic system of \eqref{eq:sys},
then 
\begin{equation} \label{eq:xkp1}
   y(t+1) = \mathcal{G}(Y,U_-)
    \begin{bmatrix} y(t)\\ u(t)  \end{bmatrix} \, ,
\end{equation}
with 
\begin{align}
\label{eq:fxu}
    \mathcal{G}(Y,U_-) =  Y_+
    \begin{bmatrix}
        Y_- \\ 
        U_-
    \end{bmatrix}^{\dagger}, 
    \end{align}
where $\dagger$ denotes the right inverse.    
\end{lemma}


\section{Robust Data-Driven Predictive Control}\label{sec:solu}

In this section, we present our proposed data-driven robust predictive control scheme. \sys{} consists of an offline data-collection phase and an online control phase which are described in the following subsections.

\subsection{Data-Collection Phase}

Due to the presence of noise, there generally exist multiple models $\begin{bmatrix}A & B \end{bmatrix}$ which are consistent with the data. As stated in Assumptions~\ref{as:noisew} and \ref{as:noisev}, we assume knowledge of the zonotopes $\mathcal{Z}_{w}$, $\mathcal{Z}_{v}$, and $\mathcal{Z}_{Av}$, that bound the noise contributions $w(t)$, $v(t)$, and $Av(t)$, and their associated matrix zonotopes $\mathcal{M}_{w}$, $\mathcal{M}_{v}$, and $\mathcal{M}_{Av}$, respectively.  %
Therefore, the goal of the data-collection phase is to construct a matrix zonotope $\mathcal{M}_\Sigma$ that over-approximates all system models consistent with the noisy data as follows. 

\begin{lemma}
\label{th:setofAB}
Given input-output trajectories $D = \{U_-, Y\}$ of the system \eqref{eq:sys}, then 
\begin{align}
    \mathcal{M}_\Sigma = (Y_{+} - \mathcal{M}_w - \mathcal{M}_v + \mathcal{M}_{Av})\begin{bmatrix} 
    Y_- \\ U_- 
    \end{bmatrix}^\dagger.
   \label{eq:zonoAB}
\end{align} 
 contains all matrices $\begin{bmatrix}A & B \end{bmatrix}$ that are consistent with the data and noise bounds. 
\end{lemma}

\begin{proof}
The proof follows the proof of \citep[Thm.1]{conf:datadriven_reach}. From the system description in \eqref{eq:sys}, we have
\begin{align} \label{eq:yp}
    Y_+ - {V}_+  = \begin{bmatrix}A  & B\end{bmatrix} \begin{bmatrix} Y_- \\ U_- \end{bmatrix} - A{V}_- + {W}_-\,,
\end{align}
where ${W}_-$, ${V}_-$, ${V}_+$ and $A{V}_-$ are the noise in the collected data. While the noise in the collected data ${W}_-$, ${V}_+$, ${V}_-$ and $A{V}_-$ is unknown, we can use the respective bounds $\mathcal{M}_w$, $\mathcal{M}_v$ and $\mathcal{M}_{Av}$ to obtain \eqref{eq:zonoAB}, where $\begin{bmatrix}A  & B\end{bmatrix} \in \mathcal{M}_\Sigma$ given that ${W}_- \in \mathcal{M}_{w}$, ${V}_+, {V}_- \in \mathcal{M}_{v}$ and $A {V}_- \in \mathcal{M}_{Av}$, thereby proving \eqref{eq:zonoAB}.
\end{proof}
\begin{remark}
  Solving the data-driven formulation above without 
    the assumption that $Av(t)$ is bounded by a known bound
    remains an open problem in the field of data-driven control.
  Notable other works such as \citep{conf:formulas} require a similar  
    assumption to derive controllers in the setting with noisy measurements.
\end{remark}
\begin{remark}
The offline data-collection phase in Lemma \ref{th:setofAB} requires that there exists a right inverse of the matrix $\begin{bmatrix} Y_- \\ U_- \end{bmatrix}$. This is equivalent to requiring this matrix to have full row rank, i.e. $\mathrm{rank} \begin{bmatrix} Y_- \\ U_- \end{bmatrix} = n+m$. This condition can be easily checked given the data. Note that for noise-free measurements this rank condition can also be enforced by choosing the input persistently exciting of order $n+1$ (compare to Lemma \ref{lm:exciting}).
\end{remark}
Next, we describe the online control phase which makes use of $\mathcal{M}_\Sigma$ obtained in the offline data-collection phase.


\subsection{Online Control Phase}

The problem we consider is receding horizon optimal control on the system~\eqref{eq:sys} with constraints in \eqref{eq:sys_con}, and where the process and measurement noise follow Assumptions~\ref{as:noisew} and \ref{as:noisev}. 
%
Since the system model is unknown, but exciting input and noisy output trajectories are available, we replace the model-based description in the traditional MPC problem by a data-driven system representation that depends on the matrix zonotope provided in Lemma \ref{th:setofAB}. More specifically, we compute the control input $u_{t+k|t}$ at each time step $t$ such that the predicted output $y_{t+k+1|t}$ stays within the computed reachable region at the next time step $t+1$ over the horizon $N$ and the cost is minimized.  
%

According to the following Lemma, we compute the reachable region from data, given the actual measured output $y(t)$ at each time step. Let $\hat{\mathcal{R}}_t$ be the data-driven reachable set and $\mathcal{R}_t$ be the model-based reachable set given the true model.

\begin{lemma}
\label{th:reach_lin}
Given input-output trajectories $D =\{U_-, Y\}$ of the system in \eqref{eq:sys}, then 
\begin{align}
\hat{\mathcal{R}}_{t+1} &= \mathcal{M}_{\Sigma} (\hat{\mathcal{R}}_{t} \times \mathcal{Z}_{u,t}  ) +  \mathcal{Z}_w + \mathcal{Z}_v - \mathcal{Z}_{Av},\label{eq:Rkp1}
\end{align}
contains the reachable set, i.e., $\hat{\mathcal{R}}_{t+1} \supseteq \mathcal{R}_{t+1}$ where $\hat{\mathcal{R}}_{0}=\zono{y(0),0}$, and $\mathcal{Z}_{u,t}=\zono{u(t),0}$. 
\end{lemma}
\begin{proof} 
We have from \eqref{eq:sys} 
\begin{align}
y(t+1) = A y(t) + B u(t) + w(t) +  v(t+1) - Av(t).
\end{align}
The reachable set computed based on the model can be found using
\begin{align}
\mathcal{R}_{t+1} &=\begin{bmatrix} A & B \end{bmatrix} (\mathcal{R}_{t} \times \mathcal{Z}_{u,t} )+  \mathcal{Z}_w+ \mathcal{Z}_v - \mathcal{Z}_{Av}.
\end{align}
Since $\begin{bmatrix} A & B \end{bmatrix} \in \mathcal{M}_{\Sigma}$ according to Lemma~\ref{th:setofAB}, both $\mathcal{R}_t$ and $\hat{\mathcal{R}}_t$ start from the same initial measured output $y(0)$, i.e., $\zono{y(0),0}$, and have the exact input sequence, i.e., $\mathcal{Z}_{u,t}=\zono{u(t),0}$, it holds that $\mathcal{R}_{t+1} \subseteq \hat{\mathcal{R}}_{t+1}$.
\end{proof}

We formulate the following data-driven optimal control problem  at time $t$. 

\begin{subequations}
\label{eq:optzonopc}
\begin{alignat}{2}
&\!\min_{u,y,s_{u},s_{l}}        &\qquad& \sum_{k=0}^{N-1} \norm{y_{t+k+1|t} -r_y(t+k+1)}_Q^2 \notag \\ 
&&&~~ +\norm{u_{t+k|t} - r_u(t+k)}_R^2 \\
&\text{s.t.} &      & \hat{\mathcal{R}}_{t+k+1|t} = \mathcal{M}_{\Sigma} (\hat{\mathcal{R}}_{t+k|t} \times \mathcal{Z}_{u,t+k} ) +\mathcal{Z}_{w}+\mathcal{Z}_{v}\nonumber\\
&&& \qquad \quad \qquad -\mathcal{Z}_{Av}, \label{eq:Rconst}\\
&                  &      &  u_{t+k|t} \in \mathcal{U}_{t+k}, \label{eq:uconst}\\
&                  &      & y_{t+k+1|t} + s_{u,t+k+1|t} = \hat{\mathcal{R}}_{u,t+k+1}, \label{eq:const_r_up}\\
&                  &      & y_{t+k+1|t} - s_{l,t+k+1|t} = \hat{\mathcal{R}}_{l,t+k+1}, \label{eq:const_r_low}\\
&                  &      & y_{t+k+1|t} + s_{u,t+k+1|t} \leq \mathcal{Y}_{u,t+k+1}, \label{eq:const_y_up}\\
&                  &      & y_{t+k+1|t} - s_{l,t+k+1|t} \geq \mathcal{Y}_{l,t+k+1},\label{eq:const_y_low} \\
&                  &      & s_{u,t+k+1|t} \geq 0, \\
&                  &      & s_{l,t+k+1|t} \geq 0, \\
&                  &      & y_{t|t} = y(t), \label{eq:y0const}
\end{alignat}
\end{subequations}
%
%
where $N$ is the time horizon,  $u=(u_{t|t},\dots,u_{t+N-1|t})$, $y=(y_{t+1|t},$ $\dots,y_{t+N|t})$ are the decision variables, and $y(t)$ is the initial measured output. The norm $\norm{u_{t+k|t} - r_u(t+k)}_R^2$ denotes the weighted norm $(u_{t+k|t} - r_u(t+k))^TR(u_{t+k|t} - r_u(t+k))$, and analogously for $\norm{y_{t+k+1|t} -r_y(t+k+1)}_Q^2$. 
The $\mathcal{Z}_{u,t+k}$ consists of the control input, i.e., $ \mathcal{Z}_{u,t+k}= \zono{ u_{t+k|t}, 0}$. The $\mathcal{Y}_{l,t+k+1}$ and $\mathcal{Y}_{u,t+k+1}$ are lower and upper bounds on the individual dimensions of the output constraint zonotope $\mathcal{Y}_{t+k+1}$, and $\hat{\mathcal{R}}_{l,t+k+1}$ and $\hat{\mathcal{R}}_{u,t+k+1}$ are lower and upper bounds on the individual dimensions of the zonotope of reachable set $\hat{\mathcal{R}}_{t+k+1|t}$, respectively. The upper and lower bounds are computed by over-approximating $\hat{\mathcal{R}}_{t+k+1|t}$ by an interval as shown in \eqref{eq:zonotointU} and \eqref{eq:zonotointL}. 

\tb{The simulated $y_{t+k+1|t}$ over the horizon differs from the true value which is however guaranteed to be inside $\hat{\mathcal{R}}_{t+k+1|t}$. Thus, we compute the upper slack variable $s_{u,t+k+1|t}=  \hat{\mathcal{R}}_{u,t+k+1} - y_{t+k+1|t}$ to account for any noise realization towards $\hat{\mathcal{R}}_{u,t+k+1}$. The computed value of the upper slack variable $s_{u,t+k+1|t}$ is then used to tighten the constraint $\mathcal{Y}_{u,t+k+1}$ and assure that  $y_{t+k+1|t}  \leq \mathcal{Y}_{u,t+k+1} - s_{u,t+k+1|t}$ and similarly for the lower bounds. In other words, 
the constraints \eqref{eq:const_r_up}, \eqref{eq:const_r_low}, \eqref{eq:const_y_up} and  \eqref{eq:const_y_low}, introduce the variables $s_{u,t+k+1|t}$ and $s_{l,t+k+1|t}$, which tighten the allowable reachable region 
according to the output constraints.} In particular, \eqref{eq:const_r_up} and \eqref{eq:const_r_low} ensure that the output $y_{t+k+1|t}$ lies within the allowable reachable region, and \eqref{eq:const_y_up} and \eqref{eq:const_y_low} ensure that the allowable reachable region fulfills the output constraints. 
This in turn restricts the choice of $\mathcal{Z}_{u,t+k}$ in \eqref{eq:Rconst} and thus implicitly tightens the input constraints. The first optimal control input ${u^*(t)=u_{t|t}}$ is then applied to the system and problem \eqref{eq:optzonopc} is solved in receding horizon fashion.

Algorithm~\ref{alg:zonopc} summarizes the data-collection and online control phase of \sys{}.

\begin{algorithm}[h!]
\textbf{Input:} Input-output data pairs $D$, reference trajectories ($r_u,r_y$), input and output constraints ($\mathcal{U}_t,\mathcal{Y}_t$), cost matrices $(Q,R)$, and initial measured output $y(t)$, time horizon $N$.
\begin{enumerate}
    \item Use the data samples $D$ to compute $\mathcal{M}_\Sigma$ in \eqref{eq:zonoAB}.
    \item Set $t \leftarrow 0$.
    \item Solve \eqref{eq:optzonopc} for time horizon $N$ to get $u^*=(u_{t|t}^*, \dots, u_{t+N-1|t}^*)$ using the output $y(t)$ as initial condition.
    \item Apply the input $u_{t|t}^*$ to the plant.
    \item Set $t \leftarrow t+1$.
    \item Return to step 3.
\end{enumerate}
  \caption{\sys{}: Zonotopic Data-Driven Predictive Control.}
  \label{alg:zonopc}
\end{algorithm}

Next, we prove the robust constraint satisfaction of the introduced formulation.

\begin{theorem}[Robust constraint satisfaction] \label{th:robconsat_zono}
Let Assumptions~\ref{as:noisew} and \ref{as:noisev} hold. Furthermore, we assume that $\mathcal{Y}_{t+k}$ represents box constraints for the output $y(t+k)$ of system~\eqref{eq:sys}.  
If problem~\eqref{eq:optzonopc} is feasible at each time step, then the closed-loop system~\eqref{eq:sys} under the controller \eqref{eq:optzonopc} robustly satisfies the constraints in \eqref{eq:sys_con} at each time step $t$ under the process and measurement noise $w(t) \in \mathcal{Z}_w$ and $v(t) \in \mathcal{Z}_v$. 
\end{theorem}
\begin{proof}
According to Lemma~\ref{th:setofAB}, the computed reachable sets $\hat{\mathcal{R}}_t$ are over-approximations of the reachable sets ${\mathcal{R}}_t$, i.e., $\hat{\mathcal{R}}_t \supseteq \mathcal{R}_t$. A control input $u(t)$ is chosen such that \eqref{eq:const_y_up}, \eqref{eq:const_y_low}, \eqref{eq:const_r_up}, and \eqref{eq:const_r_low} are satisfied which guarantees that output $y(t)$ is within the intersection between the 
 over-approximated reachable set  $\hat{\mathcal{R}}_t$ and the output constraints $\mathcal{Y}_t$ regardless of the noise instantiation.
Furthermore, the bounds $\hat{\mathcal{R}}_{l,t}$, $\hat{\mathcal{R}}_{u,t}$, $\mathcal{Y}_{l,t}$, and $\mathcal{Y}_{u,t}$ are over-approximations for the corresponding sets. Therefore, under feasibility of \eqref{eq:optzonopc}, the  constraints in \eqref{eq:const_y_low} and \eqref{eq:const_y_up} robustly guarantee constraint satisfaction of $\mathcal{Y}_t$ at each time step. 
\end{proof}

\begin{figure*}[!t]
    \centering
            \begin{subfigure}[h]{0.48\textwidth}
      \centering
        \includegraphics[scale=0.30]{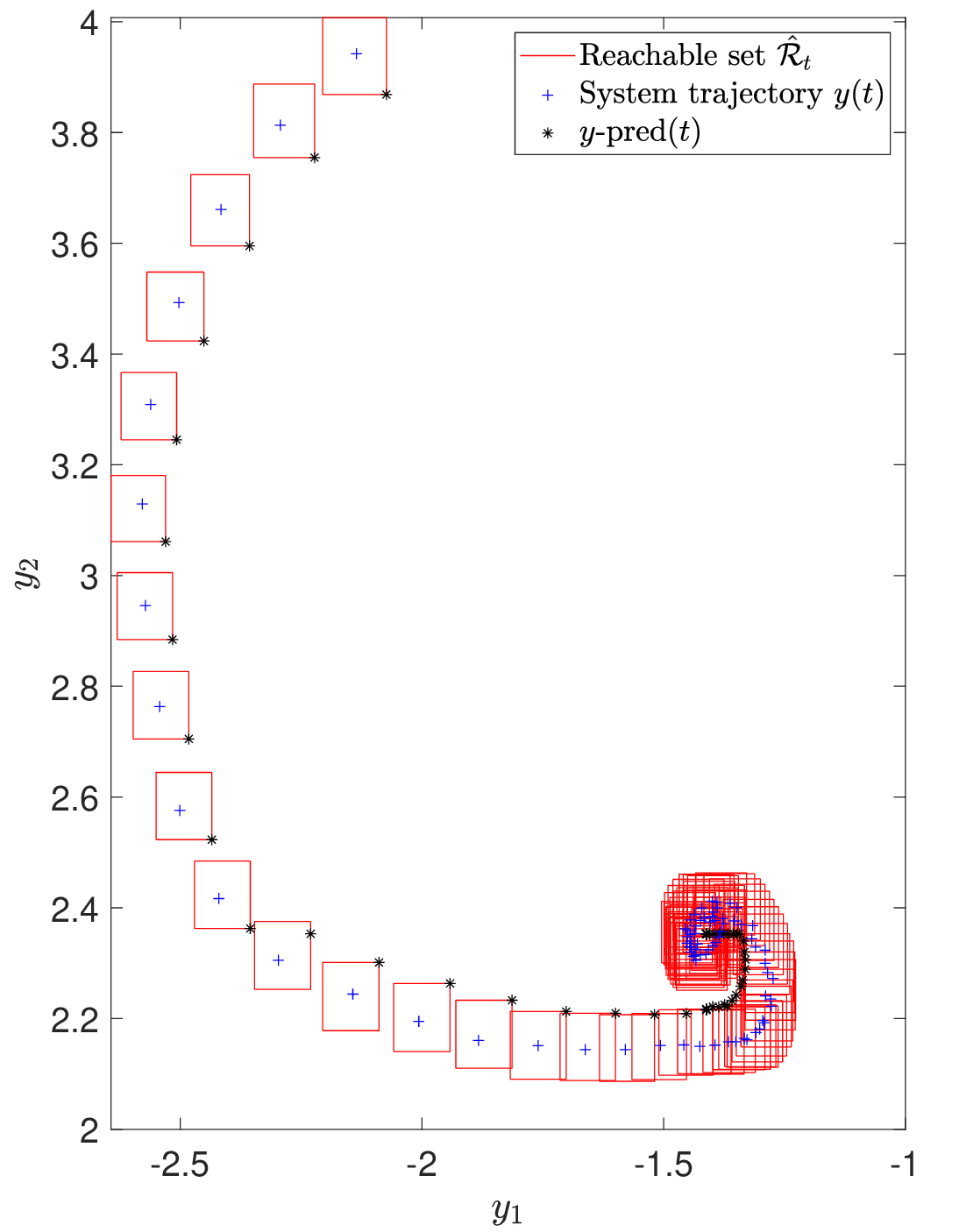}
        \caption{}
        \label{fig:y1y2_400pt}
    \end{subfigure}
            \begin{subfigure}[h]{0.48\textwidth}
      \centering
        \includegraphics[scale=0.30]{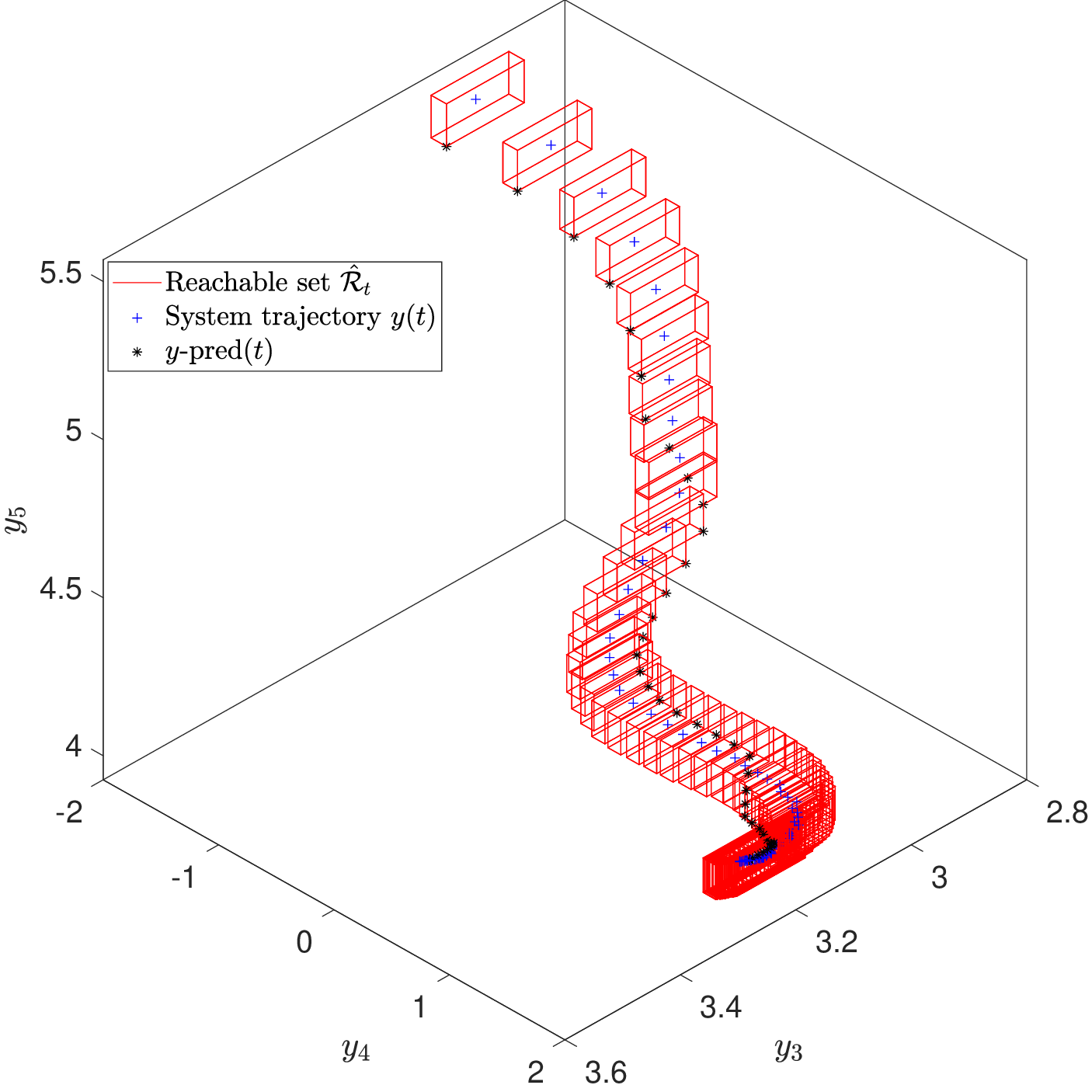}
        \caption{}
        \label{fig:y3y4_400pt}
    \end{subfigure}
\caption{The projection of the reachable sets over the time steps in the control phase with $400$ input-output pairs in the data-collection phase.}
    \label{fig:y_reach_400}
     \vspace{-2mm}
\end{figure*}

\begin{figure*}[!t]
    \centering
            \begin{subfigure}[h]{0.32\textwidth}
      \centering
        \includegraphics[width=\textwidth]{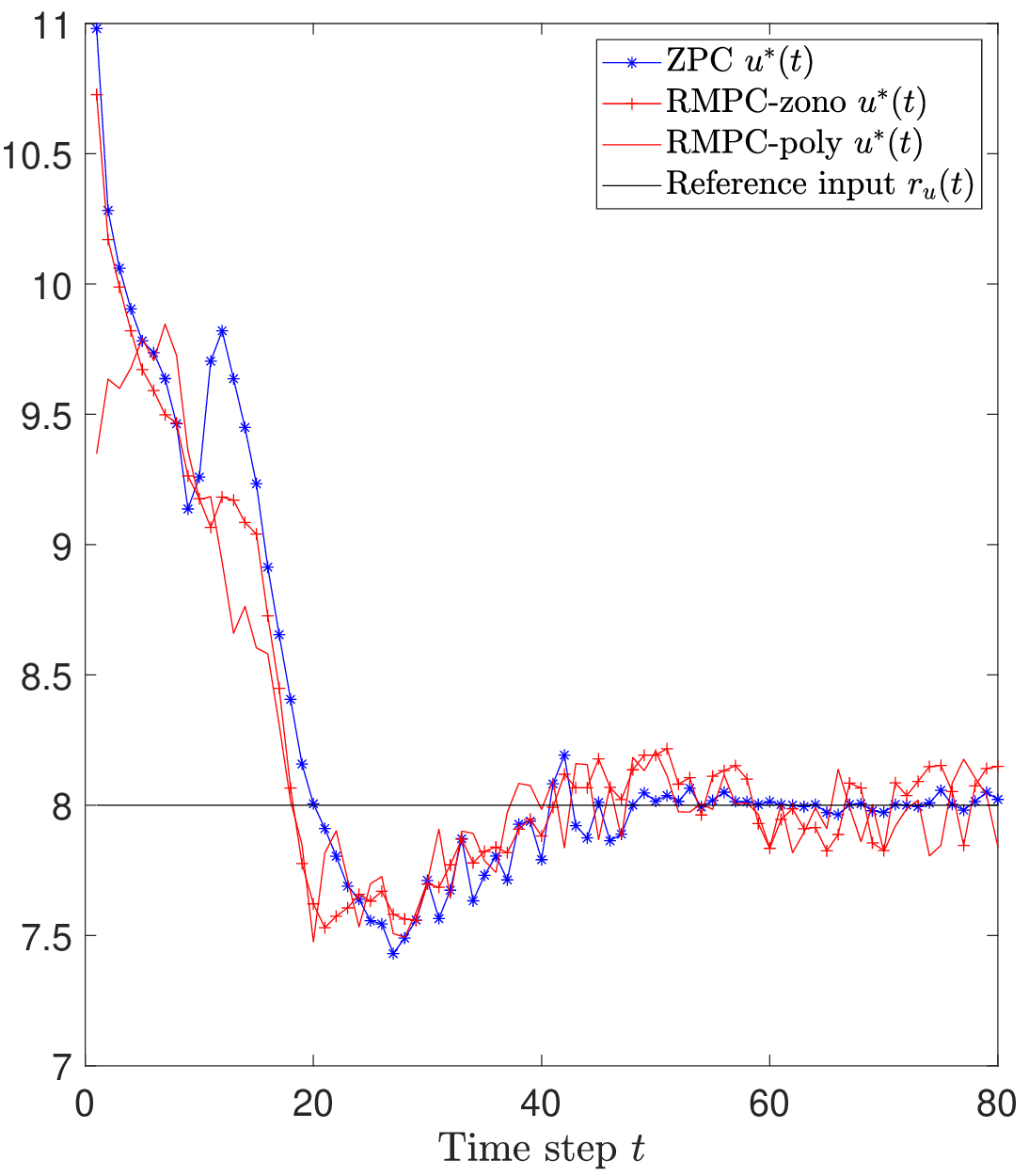}
        \caption{}
        \label{fig:upred_uref_400pt}
    \end{subfigure}
            \begin{subfigure}[h]{0.32\textwidth}
      \centering
        \includegraphics[width=\textwidth]{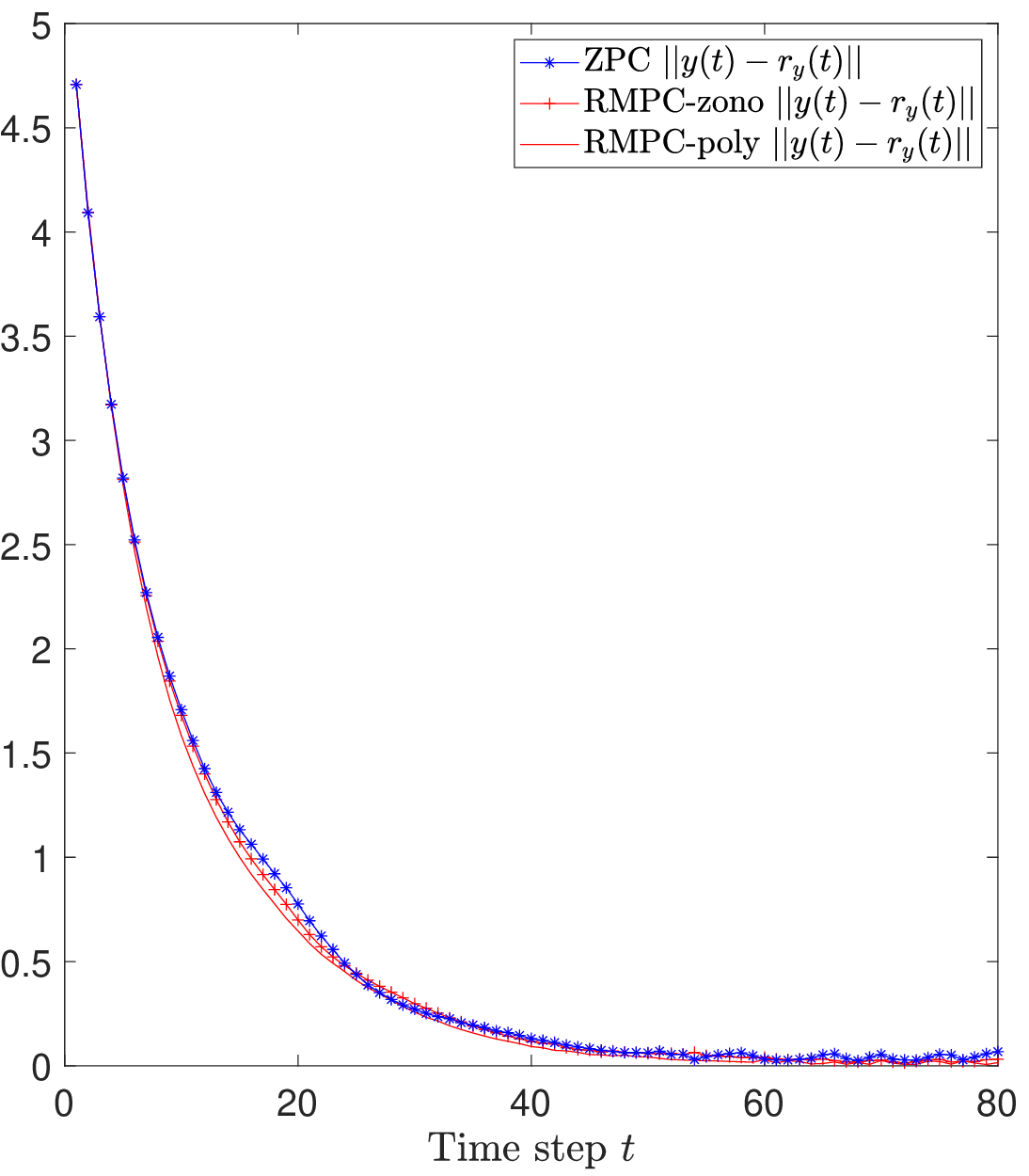}
        \caption{}
        \label{fig:yt2ref_400pt}
    \end{subfigure}
            \begin{subfigure}[h]{0.32\textwidth}
      \centering
        \includegraphics[width=\textwidth]{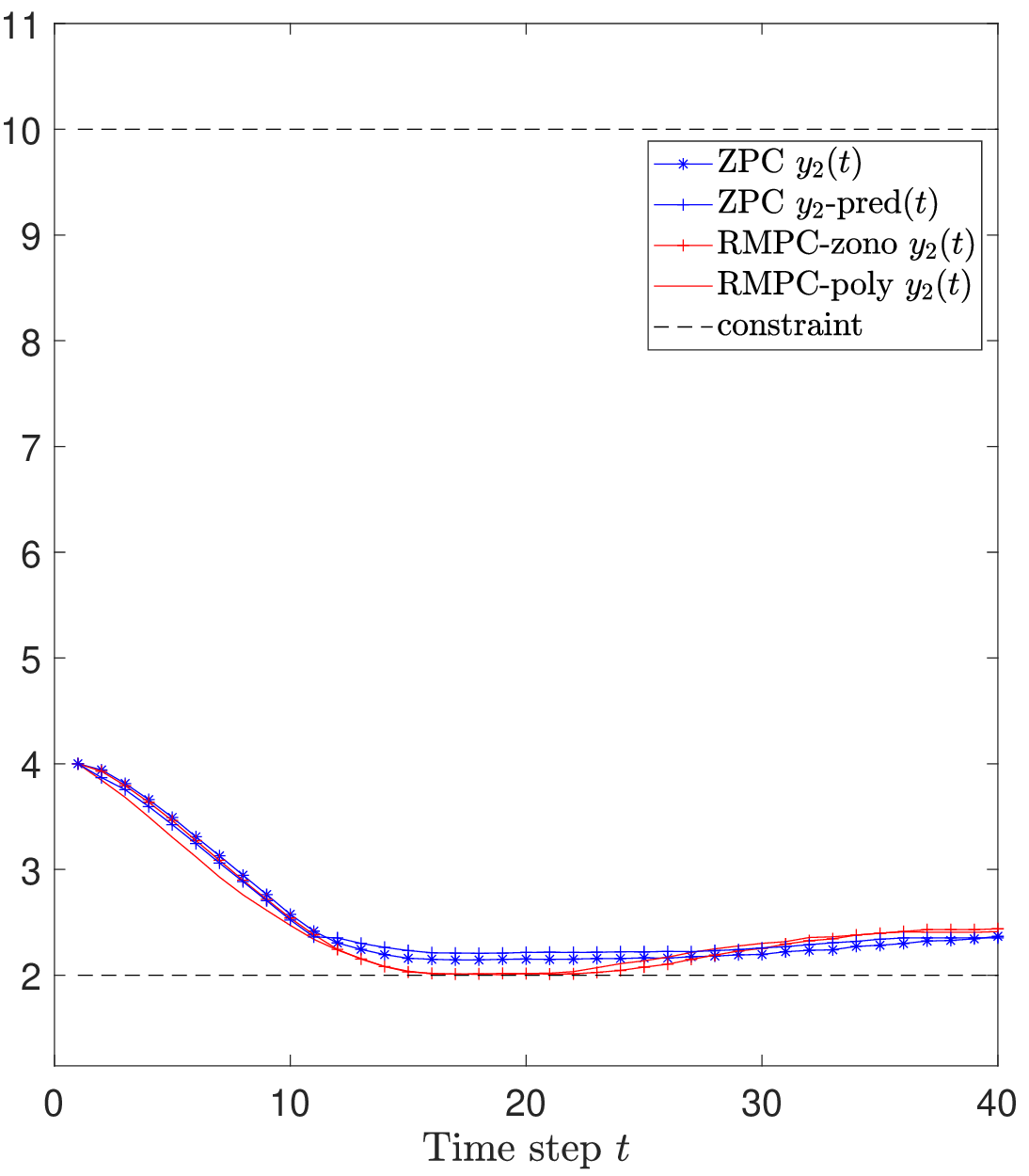}
        \caption{}
        \label{fig:robust}
    \end{subfigure}
\caption{Comparison between \sys{}, RMPC-poly, and RMPC-zono.}
    \label{fig:threeresults}
     \vspace{-3mm}
\end{figure*}

\begin{figure}[!htbp]
    \centering
    \includegraphics[scale=0.31]{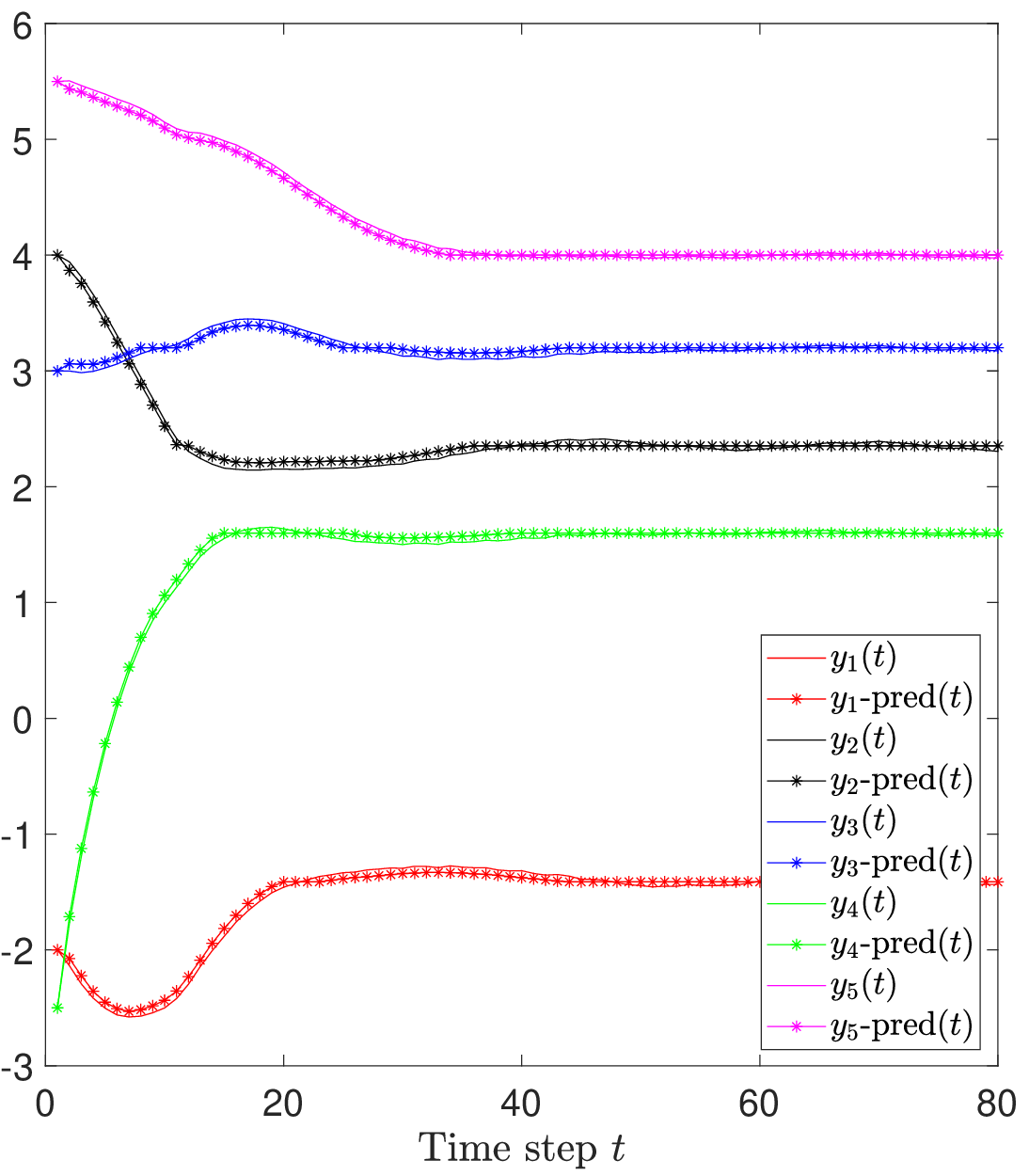}
    \caption{The predicted output and system trajectory for \sys{}.}
    \label{fig:ys_ypreds_400pt}
\end{figure}

\begin{figure*}[!t]
    \centering
            \begin{subfigure}[h]{0.45\textwidth}
      \centering
        \includegraphics[scale=0.30]{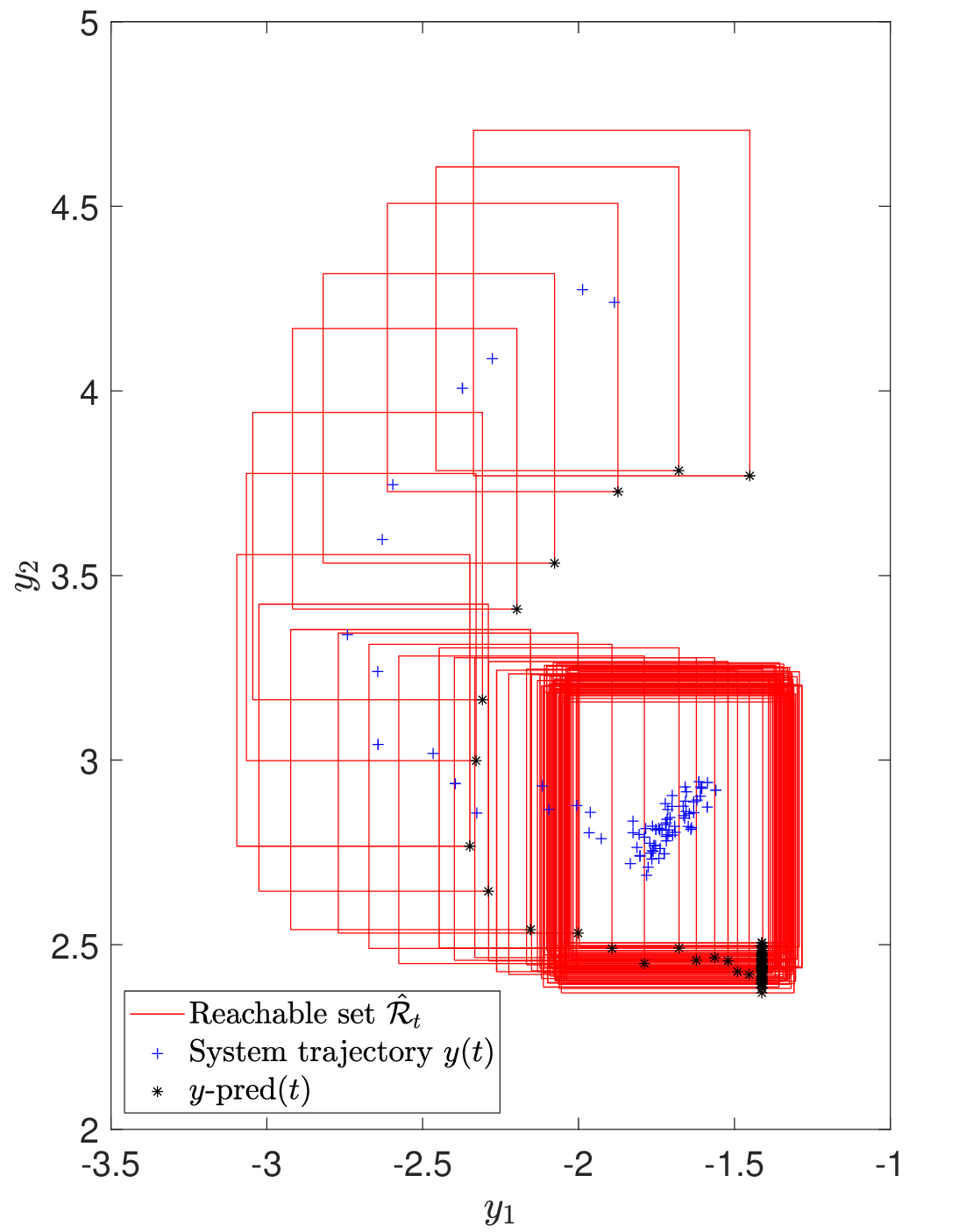}
        \caption{}
        \label{fig:y1y2_400pt_more_noise}
    \end{subfigure}
            \begin{subfigure}[h]{0.45\textwidth}
      \centering
        \includegraphics[scale=0.30]{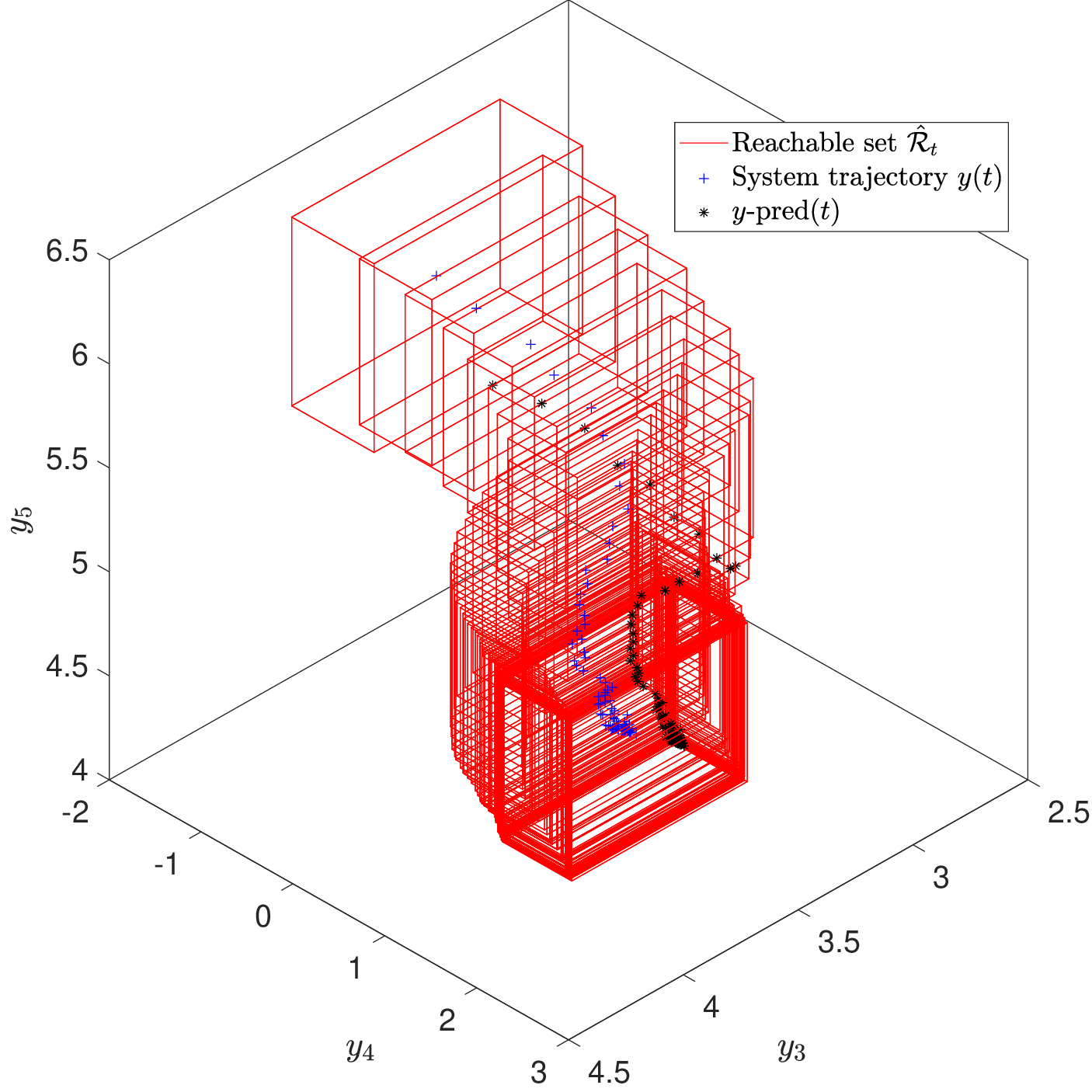}
        \caption{}
        \label{fig:y3y4_400pt_more_noise}
    \end{subfigure}
\caption{Projection of the reachable sets over the time steps with a higher magnitude of noise.}
    \label{fig:y_reach_400_morenoise}
      \vspace{-1mm}
\end{figure*}

\begin{figure*}[!t]
    \centering
            \begin{subfigure}[h]{0.32\textwidth}
      \centering
        \includegraphics[width=\textwidth]{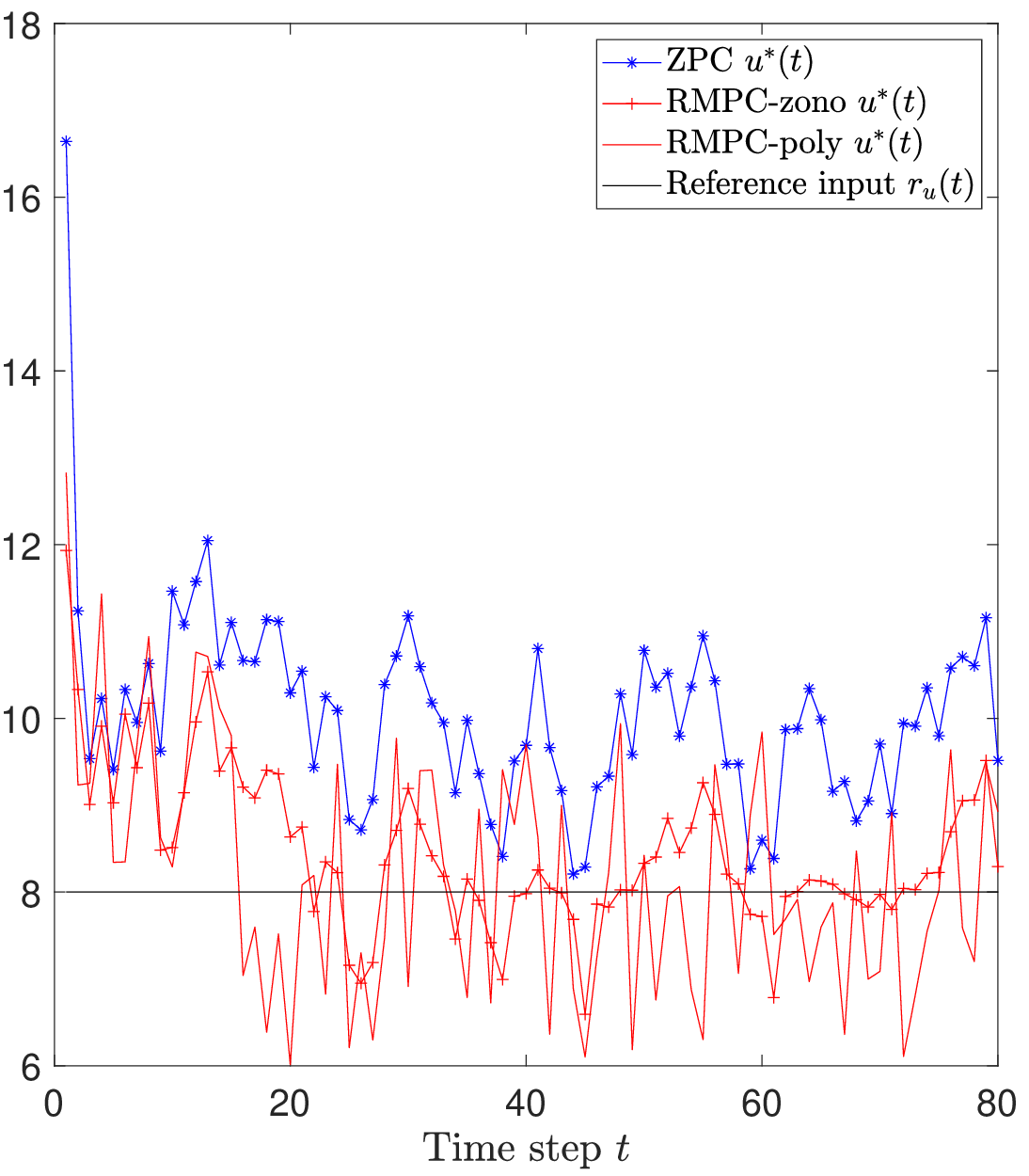}
        \caption{}
        \label{fig:upred_uref_400pt_morenoise}
    \end{subfigure}
            \begin{subfigure}[h]{0.32\textwidth}
      \centering
        \includegraphics[width=\textwidth]{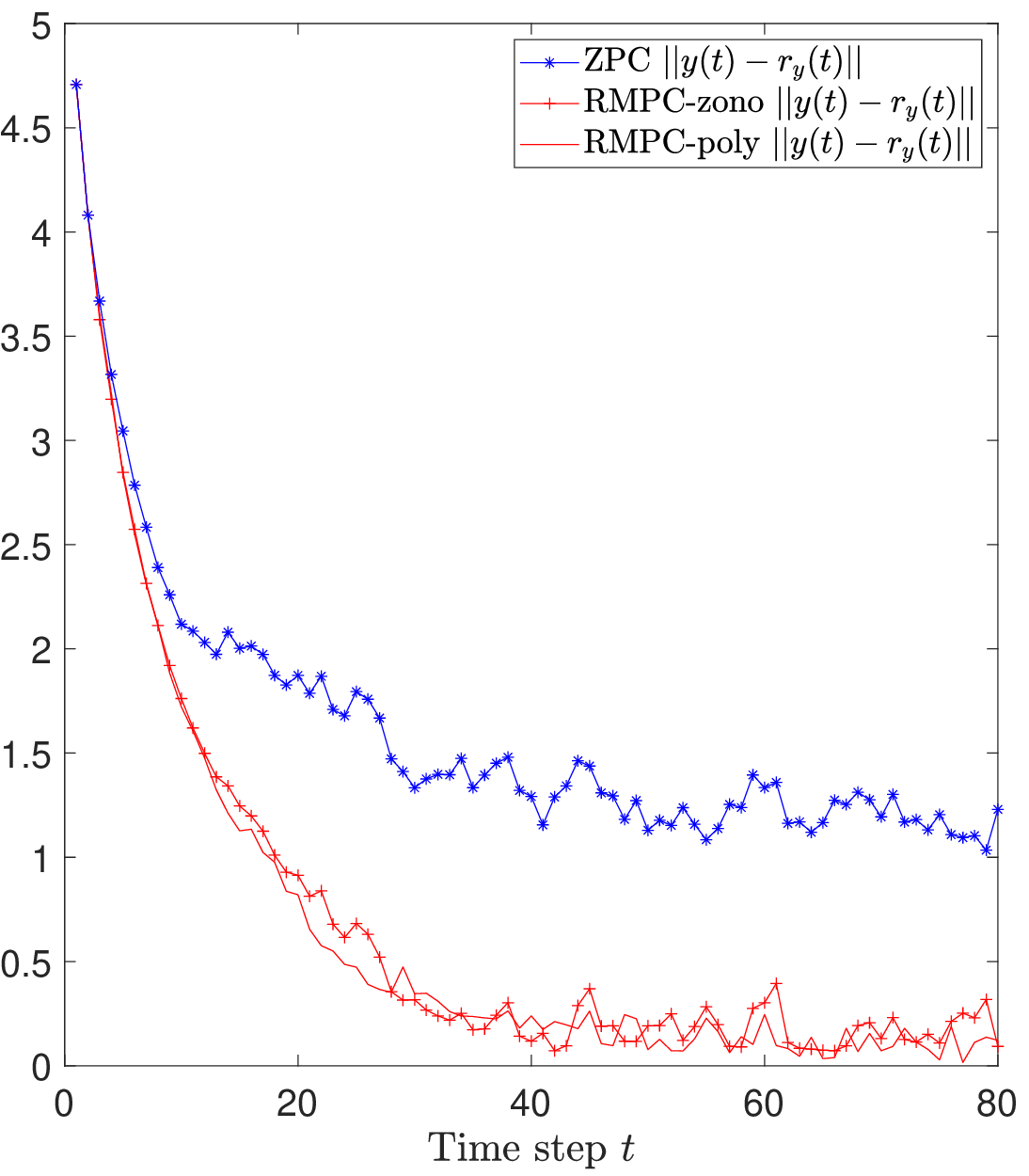}
        \caption{}
        \label{fig:yt2ref_400pt_morenoise}
    \end{subfigure}
            \begin{subfigure}[h]{0.32\textwidth}
      \centering
        \includegraphics[width=\textwidth]{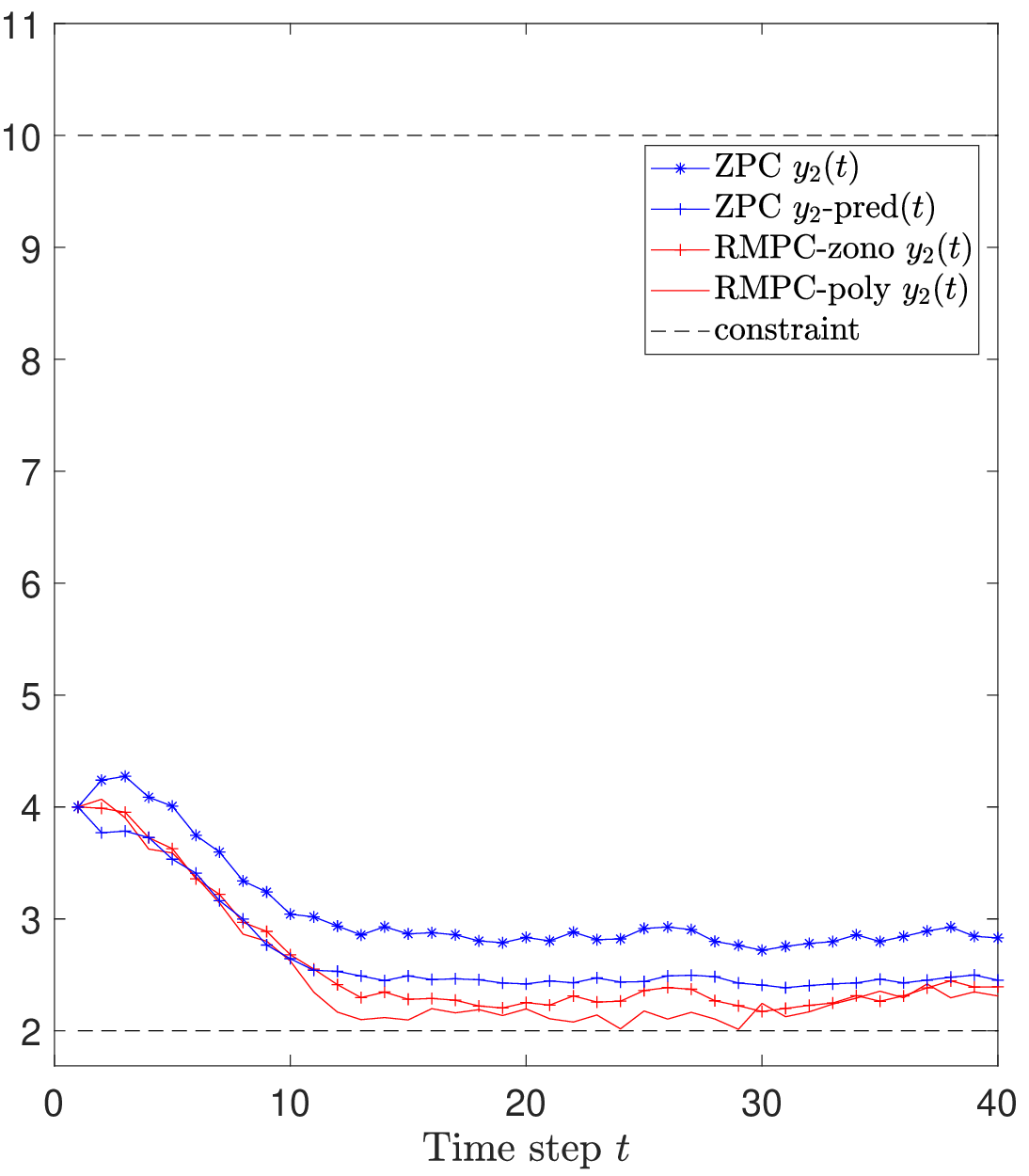}
        \caption{}
        \label{fig:addingrobust_more_noise}
    \end{subfigure}
\caption{Comparison between \sys{}, RMPC-poly, and RMPC-zono with a higher magnitude of noise.}
    \label{fig:threeresults_morenoise}
    \vspace{-4mm}
\end{figure*}

\begin{figure}[!htbp]
    \centering
    \includegraphics[scale=0.31]{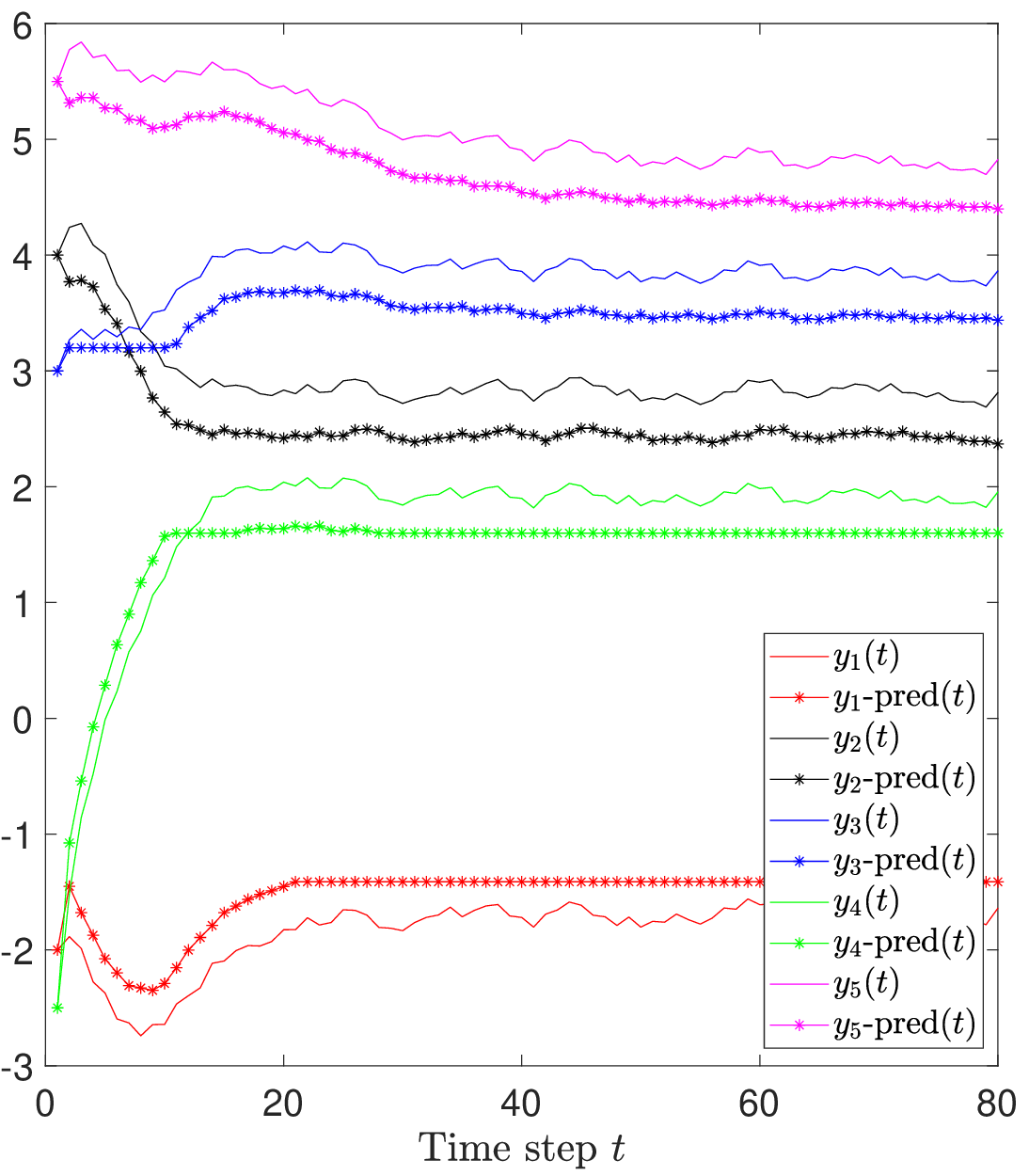}
    \caption{The predicted output and system trajectory for \sys{} in with a higher magnitude of noise.}
    \label{fig:ys_ypreds_400pt_morenoise}
\end{figure}




\subsection{Equivalence to MPC with Known Model in the Nominal Case}

The following theorem shows the equivalence between \sys{} and the nominal model-based MPC (with known model) in the noise-free case which is given by 
\begin{subequations}
\label{eq:mpc_}
\begin{alignat}{2}
&\!\min_{u,y}        &\qquad& \sum_{k=0}^{N-1} \norm{y_{t+k+1|t} -r_y(t+k+1)}_Q^2 \notag\\ 
&                  &      & ~~ +\norm{u_{t+k|t} - r_u(t+k)}_R^2 \label{eq:optProb_model_}\\
&\text{s.t.} &      & y_{t+k+1|t}= A y_{t+k|t} + Bu_{t+k|t} \\
&                  &      & u_{t+k|t} \in \mathcal{U}_{t+k},  \\
&                  &      & y_{t+k+1|t}\in \mathcal{Y}_{t+k+1}, ~\\ 
&                  &      & y_{t|t} = y(t), 
\end{alignat}
\end{subequations}
\begin{theorem}[Equivalence of \sys{} to MPC with known model in the nominal case]
 The nominal MPC with known model in \eqref{eq:mpc_}, and the \sys{} in \eqref{eq:optzonopc} result in equivalent closed-loop behaviour in the case of a noise-free system, i.e., $v(t)=0$ and $w(t)=0$, assuming that the input is persistently exciting of order $n +1$.
\end{theorem}
\begin{proof}
\tb{In the noise-free case we have 
$$\mathcal{M}_w=\mathcal{M}_v=\mathcal{M}_{Av}=0$$ 
which if inserted in \eqref{eq:zonoAB} results in:
\begin{align}
  \mathcal{M}_\Sigma = Y_+ \begin{bmatrix} Y_- \\ U_-\end{bmatrix}^\dagger  
\end{align}
which is a single model equals} to $\mathcal{G}(Y,U_-)$ in Lemma~\ref{lm:xkp1} under the assumption that rank$\Big(\begin{bmatrix} Y_-^T & U_-^T\end{bmatrix}^T\Big)$ = $n+m$ as the input is persistently exciting of order $n +1$. Given that $\mathcal{M}_\Sigma$ is a single true model, then $\hat{\mathcal{R}}_t$ is a single point equal to $x(t)$ as in the nominal MPC in \eqref{eq:mpc_}.  
\end{proof}

\section{Evaluation}\label{sec:eval}
In this section, we will compare the data-driven ZPC without model knowledge to two model-based schemes, where we assume model knowledge. We start by describing the model-based schemes, before presenting the experimental simulation results. 
\subsection{Model-Based Schemes}\label{sec:MPC}
Given the system model in \eqref{eq:sys}, the output and input reference trajectories $r_y(t) \in \mathbb{R}^p$ and $r_u(t) \in \mathbb{R}^m$, the input and output constraints $\mathcal{U}_t$ and $\mathcal{Y}_t$, the weighting matrices $R$ and $Q$, the input zonotope $\mathcal{Z}_{u,t}$, and the noise zonotopes $\mathcal{Z}_{w}$ and $\mathcal{Z}_{v}$, respectively, we implemented two MPC schemes as follows: 
\begin{itemize}
    \item RMPC-poly: A robust MPC scheme using polytopes where we transform the given input, constraints and noise zonotopes into polytopes. We then solve \eqref{eq:mpc_} using open-loop robust MPC policies with constraint tightening minimizing the nominal cost.
    \item RMPC-zono: We use the model information in \sys{}. More specifically, instead of  $\mathcal{M}_\Sigma$,   $\begin{bmatrix}
    A & B
\end{bmatrix}$ is being used in \eqref{eq:optzonopc}, minimizing the nominal cost. 
\end{itemize}

\subsection{Comparison between ZPC and the Model-Based Schemes}

To demonstrate the usefulness of \sys{}, we consider the predictive control of a five dimensional system which is a discretization of the system used in \citep[p.39]{conf:thesisalthoff} with sampling time $0.05$ sec. The discrete system dynamics are 

\begin{align*}
A&=\begin{bmatrix}
    0.9323 &  -0.1890   &      0   &      0   &      0 \\
    0.1890 &   0.9323  &       0  &       0   &      0 \\
         0 &        0  &  0.8596  &   0.0430  &        0 \\
         0 &         0   & -0.0430    & 0.8596      &    0 \\
         0 &         0  &        0    &      0   &  0.9048
\end{bmatrix}, \\
B&=\begin{bmatrix} 
    0.0436&
    0.0533&
    0.0475&
    0.0453&
    0.0476
    \end{bmatrix}^T.
\end{align*}

We make use of the CORA toolbox \citep{conf:cora} along with implementations from \citep{conf:felix} \tb{in Matlab along with the Multi-Parametric Toolbox \citep{conf:mpt} and MOSEK solver \citep{conf:mosek}}. The input set is $\mathcal{Z}_{u,t}{=}\langle 7,$ $19\rangle$. We start by computing the reachable set when there is random noise sampled from the zonotopes $\mathcal{Z}_w{=}\langle0,[0.01 ,\dots ,$ $0.01]^T\rangle$ and $\mathcal{Z}_v=\zono{0,[0.002 ,\dots , 0.002]^T}$. We collect $400$ random input-output pairs in the offline data-collection phase. We start by plotting the reachable sets, the system trajectories $y(t)$, and the predicted output $y_{t+1|t}, ... , y_{k+N+1|t}$ of \eqref{eq:optzonopc} over the time steps in Fig.~\ref{fig:y_reach_400} during the online control phase. Fig.~\ref{fig:y_reach_400} shows the system trajectory and the predicted output inside the reachable sets. 


We perform the control schemes \sys{}, RMPC-poly, and RM-PC-zono using the same realization of random noise. 
The control inputs $u(t)$ for the model based predictive control schemes and for \sys{} along with the reference input are presented in Fig.~\ref{fig:upred_uref_400pt}. We show the norm $\norm{y(t)-r_y(t)}$ for the \sys{}, RMPC-poly, and RMPC-zono in Fig.~\ref{fig:yt2ref_400pt}. Fig.~\ref{fig:yt2ref_400pt} shows that \sys{} is comparable to the RMPC-poly, and RMPC-zono given the aforementioned noise. Fig.~\ref{fig:robust} shows that the constraint $1.9 \leq y_2(t) \leq 10$ is satisfied for all control schemes. ZPC acts a bit more conservative compared to the model-based schemes due to the lack of model knowledge. Fig.~\ref{fig:ys_ypreds_400pt} shows the alignment between the system trajectory and the predicted output $y_{t+k+1|t}$ for \sys{}. \tb{The extra computation in ZPC in comparison to RMPC-zono lies in the multiplication between the matrix zonotope $\mathcal{M}_\Sigma$ and the reachable sets $\hat{\mathcal{R}}_{t+k|t} \times \mathcal{Z}_{u,t+k}$ which depends on the number of the generators of $\mathcal{M}_\Sigma$. A reduce operator can be applied on $\mathcal{M}_\Sigma$ to decrease the number of generators at the cost of over approximation. Table \ref{tab:execTime} shows the mean and standard deviation of the execution time of the three schemes. \sys{} takes around $0.399$ \SI{}{sec} on average to run in comparison to  0.057 \SI{}{sec} and 0.196 \SI{}{sec} for RMPC-poly and RMPC-zono, respectively. }  


Next, we consider noise of a magnitude 10 times as high as in the previous test case. More specifically, we consider $\mathcal{Z}_w=\zono{0,\begin{bmatrix}0.1 ,\dots , 0.1\end{bmatrix}^T}$ and $\mathcal{Z}_v=\zono{0,\begin{bmatrix}0.02 ,\dots , 0.02\end{bmatrix}^T}$. We start by plotting again the reachable sets in Fig.~\ref{fig:y_reach_400_morenoise} which are more conservative than before. Note that having a higher magnitude of noise in the data increases the number of possible models in $\mathcal{M}_\Sigma$ which in turns increases the size of the reachable sets and affects the overall performance. The control inputs $u(t)$ for the model based predictive control schemes and \sys{} along with the reference input are presented in Fig.~\ref{fig:upred_uref_400pt_morenoise}. We plot $\norm{y(t)-r_y(t)}$ in Fig.~\ref{fig:yt2ref_400pt_morenoise} with again $400$ data points in the data-collection phase. Fig.~\ref{fig:addingrobust_more_noise} shows the effect of a higher magnitude of noise on satisfying the constraint~$1.9~\leq~y_2(t)~\leq~10$. Fig.~\ref{fig:ys_ypreds_400pt_morenoise} shows the system trajectory and the predicted output $y_{t+k+1|t}$ for \sys{}.

\section{Conclusion}\label{sec:conc}
We propose a zonotopic data-driven predictive control sche-me named \sys{}. Our proposed controller consists of two phases: 1) an offline data-collection phase, during which a matrix zonotope is learned from data as a data-driven system representation, and 2) an online control phase. During the online control phase, we compute data-driven reachable sets based on a matrix zonotope recursion. In the noise-free case, \sys{} is equivalent to a nominal MPC scheme. In the case of process and measurement noise \sys{} provides robust constraint satisfaction. We show the effectiveness of the data-driven control scheme \sys{} in numerical experiments compared to two model-based predictive control schemes. \tb{For future work, we will consider guaranteeing recursive feasibility and nonlinear systems. Furthermore, we will consider testing on high dimensional systems.}

\begin{table}[t]
\centering 
\caption{\tb{The mean and standard deviation of the execution time of RMPC-poly, and RMPC-zono, and \sys{} in \SI{}{sec}.}} 
\label{tab:execTime}
\centering 
\normalsize
\tb{\begin{tabular}{c  c c c c c}
\toprule
Scheme & Mean & Std \\ 
\midrule
RMPC-poly & 0.057 & 0.028 \\
RMPC-zono &0.196 & 0.050\\
\sys{}  &0.399 & 0.090 \\
\bottomrule
\end{tabular}}
\end{table}

\section*{Acknowledgement}
This work was supported by the Swedish Research Council, the Knut and Alice Wallenberg Foundation, the Democritus project on Decision-making in Critical Societal Infrastructures by Digital Futures, the European Unions Horizon
   2020 Research and Innovation programs under the CONCORDIA cyber security project (GA No. 830927) and the Marie Sklodowska-Curie grant agreement No.\ 846421. 
\bibliography{references}

\end{document}